\documentclass[11pt]{article}

\usepackage{amsmath,amssymb,amsthm}
\usepackage{booktabs}
\usepackage{hyperref}
\usepackage[margin=1in]{geometry}

\newtheorem{theorem}{Theorem}

\title{A Simple Constructive Bound on Circuit Size Change\\Under Truth Table Perturbation}
\author{Kirill Krinkin\\
Neapolis University Pafos\\
\texttt{kirill@krinkin.com}}
\date{}

\begin{document}
\maketitle

\begin{abstract}
The observation that optimum circuit size changes by at most $O(n)$ under a
one-point truth table perturbation is implicit in prior work on the Minimum
Circuit Size Problem. This note states the bound explicitly for arbitrary fixed
finite complete bases with unit-cost gates, extends it to general Hamming
distance via a telescoping argument, and verifies it exhaustively at $n = 4$ in
the AIG basis using SAT-derived exact circuit sizes for 220 of 222 NPN
equivalence classes. Among 987 mutation edges, the maximum observed difference
is $4 = n$, confirming the bound is tight at $n = 4$ for AIG.
\end{abstract}

\section{Introduction}

The observation that optimum circuit size changes by at most $O(n)$ under a
one-point perturbation is already implicit in work on the Minimum Circuit Size
Problem (MCSP), where it serves as a bounded-differences ingredient in
concentration arguments~\cite{golovnev2019mcsp}. This note isolates the
phenomenon as a standalone statement for arbitrary fixed finite complete bases,
formulates the resulting Hamming-distance bound explicitly, and complements it
with exact small-scale evidence in the AIG model.

The optimal circuit size $\mathrm{opt}(f, B)$ of a Boolean function $f$ over a
complete gate basis $B$ is a fundamental measure of computational complexity.
A natural question arises: how sensitive is this measure to small changes in the
function?

We observe that changing a single bit of a Boolean function's truth table
changes the optimal circuit size by at most $O(n)$ gates, where $n$ is the
number of variables. The bound is constructive: we exhibit the modified circuit
explicitly. We verify the bound exhaustively for $n = 4$ in the AIG basis using
SAT-verified exact circuit sizes, and show that the bound is tight at $n = 4$ in the AIG basis.
The exhaustive verification is restricted to $n = 4$ and one circuit model; it
should be read as validation of the construction and local tightness, not as
asymptotic evidence.

\section{Preliminaries}

A \emph{Boolean function} of $n$ variables is a map $f \colon \{0,1\}^n \to
\{0,1\}$. Its \emph{truth table} $\mathrm{tt}(f)$ is the binary string of
length $2^n$ listing the output values in lexicographic order of inputs.

The \emph{Hamming distance} $d_H(\mathrm{tt}(f), \mathrm{tt}(f'))$ between two
truth tables is the number of positions at which they disagree.

A \emph{circuit} over a gate basis $B$ is a directed acyclic graph whose
internal nodes are labelled by gates from $B$ and whose sources are the input
variables $x_0, \ldots, x_{n-1}$. The \emph{size} of a circuit is the number
of internal nodes (gates). The \emph{optimal circuit size} $\mathrm{opt}(f, B)$
is the minimum size over all circuits in basis $B$ that compute
$f$~\cite{jukna2012boolean}.

A basis $B$ is \emph{complete} if every Boolean function can be computed by a
circuit over $B$. Throughout this note, we assume $B$ is finite and each gate
has unit cost. We consider two concrete bases:
\begin{itemize}
  \item \textbf{AIG basis:} $\{\mathrm{AND}(x,y), \neg x\}$ with free
    inversions---every gate output is available in both polarities at no cost.
    Circuit size counts only AND gates.
  \item \textbf{AON basis:} $\{\mathrm{AND}(x,y), \mathrm{OR}(x,y), \neg x\}$.
    Circuit size counts AND, OR, and NOT gates.
\end{itemize}

\section{Main Result}

\begin{theorem}\label{thm:lipschitz}
For any fixed finite complete gate basis $B$ with unit-cost gates and any two Boolean functions $f, f'$ of $n$
variables,
\[
  |\mathrm{opt}(f, B) - \mathrm{opt}(f', B)| \;\le\; c_B \cdot n \cdot
  d_H(\mathrm{tt}(f), \mathrm{tt}(f')),
\]
where $c_B$ is a constant depending only on $B$.
\end{theorem}

\begin{proof}
It suffices to prove the case $d_H = 1$. For general $d_H$, let
$f = f_0, f_1, \ldots, f_{d_H} = f'$ be a sequence of functions where each
consecutive pair differs in exactly one truth table bit. Applying the $d_H = 1$
bound at each step and telescoping:
\[
  |\mathrm{opt}(f, B) - \mathrm{opt}(f', B)| \;\le\;
  \sum_{i=1}^{d_H} |\mathrm{opt}(f_{i-1}, B) - \mathrm{opt}(f_i, B)|
  \;\le\; c_B \cdot n \cdot d_H.
\]

Suppose $f$ and $f'$ differ only at input $x^* \in \{0,1\}^n$, with
$f(x^*) = 0$ and $f'(x^*) = 1$. (The case $f(x^*) = 1$ is symmetric by
swapping $f$ and $f'$.)

\medskip
\noindent\textbf{Upper bound on $\mathrm{opt}(f')$.}
Given an optimal circuit $S$ for $f$ of size $\mathrm{opt}(f)$, construct a
circuit for $f'$ as follows.

\emph{Step 1: Equality detector.} Build a subcircuit
$\mathrm{eq}(x, x^*)$ that outputs 1 if and only if $x = x^*$. This requires
checking $x_i = x^*_i$ for each $i \in \{0, \ldots, n-1\}$ and AND-ing the
results. In the AIG basis, each comparison is a literal (free inversion if
$x^*_i = 0$), so the AND chain costs $n - 1$ gates. In a general complete
basis $B$, implementing AND and NOT costs $c_B$ gates each, giving a detector
of size at most $c_B \cdot (n - 1)$. (For any fixed finite complete basis $B$, there exist constants $a_B, b_B > 0$
such that the repair gadget---detector plus output correction---has size at most
$a_B n + b_B$ gates. For $n \ge 1$ this is at most $c_B \cdot n$ where
$c_B := a_B + b_B$.)

\emph{Step 2: Output correction.} Compute
\[
  f'(x) = f(x) \lor \mathrm{eq}(x, x^*).
\]
In the AIG basis: $f'(x) = \neg(\neg f(x) \land \neg\mathrm{eq}(x, x^*))$,
costing 1 AND gate (inversions free). In a general basis, the OR costs at most
$c_B$ gates.

\emph{Total:} $\mathrm{opt}(f') \le \mathrm{opt}(f) + c_B \cdot n$.

\medskip
\noindent\textbf{Upper bound on $\mathrm{opt}(f)$.}
Given an optimal circuit $S'$ for $f'$, construct:
\[
  f(x) = f'(x) \land \neg\mathrm{eq}(x, x^*).
\]
This is correct: when $x \ne x^*$, $\mathrm{eq} = 0$ and $f(x) = f'(x)$; when
$x = x^*$, $f(x) = f'(x^*) \land 0 = 0 = f(x^*)$. The detector is reused;
the AND gate costs 1 gate in AIG (inversion of $\mathrm{eq}$ is free) or
$c_B$ gates in general.

\emph{Total:} $\mathrm{opt}(f) \le \mathrm{opt}(f') + c_B \cdot n$.

\medskip
Combining both inequalities: $|\mathrm{opt}(f) - \mathrm{opt}(f')| \le
c_B \cdot n$.
\end{proof}

\medskip
\noindent\textbf{Remark.} For the AIG basis, the detector costs $n - 1$ AND
gates (inversions are free) and the output correction costs 1 gate, giving
$c_B = 1$ and a bound of $n \cdot d_H$. For a general basis $B$, the constant
$c_B$ depends on the cost of simulating AND and NOT in $B$.

\section{Tightness}

The bound of Theorem~\ref{thm:lipschitz} is achieved at $n = 4$ in the AIG
basis. NPN class \texttt{0x0001} (the single-minterm function, one AND of all
inputs) has $\mathrm{opt}_{\mathrm{AIG}} = 3$, while NPN class \texttt{0x0180}
has $\mathrm{opt}_{\mathrm{AIG}} = 7$. These classes are connected by a
single-bit truth table mutation, giving
$|\mathrm{opt}(f) - \mathrm{opt}(f')| = 4 = n$.

An open question is whether there exists a constant $C$ independent of $n$ such
that $|\mathrm{opt}(f) - \mathrm{opt}(f')| \le C$ for $d_H = 1$. Empirical
evidence at $n = 4$ suggests the typical perturbation is much smaller than the
worst case: 94.7\% of mutation edges have $|\Delta\mathrm{opt}| \le 2$, and the
mean absolute difference is 1.03. No asymptotic lower-bound family is known;
whether $\Omega(n)$ one-bit gaps exist for infinitely many $n$ remains open.

\section{Computational Verification}

We verified the bound exhaustively at $n = 4$ for the AIG basis.
Optimal circuit sizes for all 222 NPN equivalence classes of 4-variable Boolean
functions were computed: 220 of 222 classes have SAT-verified exact values:
136 by exhaustive enumeration~\cite{krinkin2026unitgap}, 84 by SAT-based exact
synthesis~\cite{kojevnikov2009circuit} with the cube-and-conquer
strategy~\cite{heule2012cube}. The remaining 2 classes have improved upper
bounds (SAT timeout at $k = 9$, confirmed SAT at $k = 10$).

The \emph{mutation graph} connects two NPN classes by an edge whenever some
function in one class differs from some function in the other by a single truth
table bit. Among the 987 mutation edges where both endpoints have exact
$\mathrm{opt}_{\mathrm{AIG}}$ values, the maximum observed difference is
$4 = n$, confirming the bound is tight at $n = 4$ for the AIG basis.

\begin{table}[h]
\centering
\begin{tabular}{crr}
\toprule
$|\Delta\mathrm{opt}_{\mathrm{AIG}}|$ & Edges & \% \\
\midrule
0 & 300 & 30.4 \\
1 & 414 & 41.9 \\
2 & 221 & 22.4 \\
3 & 45 & 4.6 \\
4 & 7 & 0.7 \\
\midrule
Total & 987 & 100.0 \\
\bottomrule
\end{tabular}
\caption{Distribution of $|\Delta\mathrm{opt}_{\mathrm{AIG}}|$ across mutation
  edges between NPN-4 classes with exact optimal sizes.}
\label{tab:distribution}
\end{table}

The seven edges achieving $|\Delta\mathrm{opt}| = n = 4$ all connect a
function of optimal size 3 to one of optimal size 7.

\section{Discussion}

The bound implies that nearby functions (in Hamming distance) have provably
nearby circuit complexity.

\textbf{Complexity estimation.} As a theoretical corollary, the bound gives a
worst-case transfer inequality: $\mathrm{opt}(f') \in [\mathrm{opt}(f) -
c_B n,\; \mathrm{opt}(f) + c_B n]$ for any Hamming neighbour $f'$ of $f$.

\textbf{Empirical tightness.} At $n = 4$, the mean absolute difference across
mutation edges is 1.03, far below the worst-case bound of $n = 4$. Whether
this gap persists at larger $n$ is an open question. Specifically: does there
exist an explicit family of function pairs $(f_n, f'_n)$ with
$d_H(\mathrm{tt}(f_n), \mathrm{tt}(f'_n)) = 1$ and
$|\mathrm{opt}(f_n) - \mathrm{opt}(f'_n)| = \Omega(n)$?

\textbf{Basis dependence.} The constant $c_B$ reflects the overhead of
implementing the equality detector and output correction in basis $B$. For
bases with free inversions (AIG), $c_B = 1$; for bases without, $c_B$ grows
with the cost of implementing AND and NOT. Whether the linear dependence on $n$
can be improved to sublinear for specific bases remains open.

\subsection*{Data availability}
The mutation graph, optimal circuit sizes for all 222 NPN-4 classes, and
verification scripts are available at
\url{https://github.com/krinkin/bounds} (tagged \texttt{v1.0} at time of
submission).

\bibliographystyle{plain}

\end{document}